\documentclass[10pt, conference, letterpaper]{IEEEtran}

\IEEEoverridecommandlockouts

\usepackage{amsthm}
\usepackage{tikz}
\usetikzlibrary{calc,decorations.pathreplacing,arrows.meta}
\usepackage{color}
\usepackage{epsfig}

\usepackage{algorithmicx,algorithm}
\usepackage[noend]{algpseudocode}
\usepackage{adjustbox}
\usepackage{amsmath}
\usepackage{amssymb}
\usepackage{times,graphicx,amsfonts}
\usepackage{bm}
\usepackage{cases}
\usepackage{enumitem}

\newtheorem{lemma}{Lemma}

\newtheorem{theorem}{Theorem}

\newtheorem{corollary}{Corollary}

\usepackage{cite}
\usepackage{subcaption}

\usepackage{cuted}

\begin{document}
	
	\title{Stochastic Geometry of Cylinders: Characterizing Inter-Nodal Distances for 3D UAV Networks}
    
\author{
\IEEEauthorblockN{Yunfeng Jiang}
\IEEEauthorblockA{\textit{School of Mathematical Sciences} \\
\textit{Beijing Normal University}\\
Beijing, China \\
202211130045@mail.bnu.edu.cn}
\and
\IEEEauthorblockN{Zhiming Huang}
\IEEEauthorblockA{\textit{Department of Computer Science} \\
\textit{University of Victoria}\\
BC, Canada \\
zhiminghuang@uvic.ca}
\and
\IEEEauthorblockN{Jianping Pan}
\IEEEauthorblockA{\textit{Department of Computer Science} \\
\textit{University of Victoria}\\
BC, Canada \\
pan@uvic.ca}

}
	\maketitle
	
	\begin{abstract}
		The analytical characterization of coverage probability in finite three-dimensional wireless networks has long remained an open problem, hindered by the loss of spatial independence in finite-node settings and the coupling between link distances and interference in bounded geometries. This paper closes this gap by presenting the first exact analytical framework for coverage probability in finite 3D networks modeled by a binomial point process within a cylindrical region. To bypass the intractability that has long hindered such analyses, we leverage the independence structure, convolution geometry, and derivative properties of Laplace transforms, yielding a formulation that is both mathematically exact and computationally efficient. Extensive Monte Carlo simulations verify the analysis and demonstrate significant accuracy gains over conventional Poisson-based models. The results generalize to any confined 3D wireless system, including aerial, underwater, and robotic networks.
	\end{abstract}
	
	
	\section{Introduction}
The study of large-scale multi-agent systems, ranging from autonomous vehicle fleets to robotic swarms and \emph{unmanned aerial vehicles (UAVs)}, has brought renewed attention to the fundamental problem of analyzing and optimizing their collective communication performance. At the heart of this problem lies an open question that has remained unresolved for more than a decade:

\emph{Can we obtain a \textbf{precise analytical characterization of coverage probability} for a \textbf{finite, three-dimensional~(3D)} wireless network where nodes are randomly distributed within a bounded region?}

This question is central to understanding how connectivity, reliability, and interference jointly scale in cooperative autonomous systems. A rigorous solution would bridge the long-standing gap between stochastic asymptotic models, which assume infinite spatial extent, and simulation-based studies, which lack analytical generality and physical interpretability.

However, resolving this question is particularly challenging. The finite-node constraint destroys the spatial independence assumed in classical Poisson models, while three-dimensional bounded geometries introduce strong coupling between the serving distance and aggregate interference. These issues are further compounded by realistic fading and mobility models, making exact analysis notoriously intractable.

Existing research has made progress in special cases~\cite{3,4,5,6,7}. Most works model UAV networks using cylindrical or planar geometries, focusing primarily on air-to-ground links. For example, when a ground user connects to its nearest UAV or when UAVs are confined to a 2D plane. Others simplify the problem by fixing one UAV or ignoring height variations altogether. As a result, the inter-UAV distance distribution and coverage probability in finite 3D cylindrical spaces remain analytically unsolved, leaving a key theoretical and practical gap in understanding aerial-to-aerial network performance.

We resolve this long-standing open question by deriving, for the first time, the exact analytical characterization of the coverage probability for finite three-dimensional wireless networks with randomly distributed nodes inside a bounded cylindrical region. Our key technical innovation is a complete geometric–probabilistic framework that decomposes the analysis into four tractable steps:
\begin{enumerate}
    \item deriving the exact distance distribution between two uniformly random nodes within a finite 3D cylinder;
    \item obtaining the serving-link distribution via order statistics;
    \item modeling the aggregate interference through its conditional Laplace transform; and
    \item combining these components using the tower law of expectation to yield a closed-form integral expression for coverage probability.
\end{enumerate}
To bypass the intractability that has long hindered finite-node 3D analysis, we leverage the independence structure, convolution geometry, and derivative properties of Laplace transforms, resulting in an expression that is both mathematically exact and computationally efficient. Although motivated by UAV swarm networks, the theoretical results are general and directly applicable to any finite 3D wireless system, such as aerial, underwater, or robotic networks, where nodes are randomly deployed in confined spatial domains.

Furthermore, we conduct extensive Monte Carlo simulations to corroborate the analytical results, demonstrating near-perfect agreement between theory and experiment across diverse geometric configurations. These validations confirm both the accuracy and generality of the proposed framework, which consistently outperforms the conventional methods by correctly accounting for finite-node and boundary effects.

The remainder of this paper is organized as follows. Section II reviews related work. Section III introduces the system model and problem formulation. Section IV presents the main analytical framework, leading to the main results. Section V validates the analysis through numerical results and compares the proposed model with the conventional baseline. Section VI concludes the paper and outlines directions for future research.

	\section{Related works}
    In the 2D domain, the work of \cite{9,10} was highly influential. By providing the exact distance distribution for points within an arbitrary triangle, their result enabled precise performance analysis of networks deployed in realistic, irregularly shaped areas. The profound impact of this 2D result has directly motivated the push for equivalent analysis in 3D space, particularly for UAV networks.

	Modeling 3D UAV networks is crucial for performance analysis, with a significant body of work focusing on metrics like coverage probability. Among various spatial models, the cylindrical configuration is frequently adopted due to its resemblance to practical deployment scenarios. Several studies leverage this geometry to analyze air-to-ground communication links. For example, the authors of \cite{3} investigated the coverage probability for a ground user connecting to the closest UAV under a mixed mobility policy, while the authors of \cite{6} analyzed a similar scenario with UAVs deployed on the cylinder's top surface and a receiver at its base.
	
	While these works provide valuable insights into UAV-to-ground coverage, the characterization of the distance distribution and resulting coverage between aerial users within a 3D cylinder remains a key challenge. In terms of modeling the spatial distribution of UAVs, point processes are a common tool. The authors in \cite{7} utilized a multi-binomial point process to analyze UAV swarms on a 2D surface. Other research explores different topologies and complex scenarios, such as the truncated octahedron-based 3D cell model in \cite{4} or the coexistence of UAVs with \emph{device-to-device (D2D)} networks in \cite{5}. They either ignored the height of the UAVs and modeled them in 2D, or fixed a UAV in advance to simplify the analysis. Neither of them can truly describe a randomly distributed swarm of UAVs. Motivated by the gaps in existing literature, this paper focuses on deriving the fundamental inter-UAV distance distribution within a 3D cylinder, which is essential for analyzing the performance of aerial-to-aerial communication links.

\section{System Model and Problem Formulation}
\subsection{System Model}
We consider a UAV network, as shown in Fig.~\ref{fig:cylinder_uav_final}, consisting of $N$ nodes (i.e., UAVs) deployed within a three-dimensional cylindrical volume $\mathcal{C} \subset \mathbb{R}^3$ of height $H$ and base radius $R$. Node locations are \emph{independent and identically distributed (i.i.d.)} according to a uniform distribution over the volume of $\mathcal{C}$, forming a \emph{binomial point process (BPP)} that represents a finite and uniformly random deployment.

Each UAV is capable of both transmitting and receiving information, enabling peer-to-peer communication within the network.
In the following, the terms transmitter and receiver refer to the UAV that is currently transmitting or receiving information, respectively.

The wireless channel between two UAVs is characterized by both distance-dependent path loss and small-scale fading.
The received power decays with distance $d$ as $d^{-\alpha}$, where $\alpha > 2$ is the {path loss exponent}. The fading gain $G$ follows a Nakagami-$m$ distribution, $G \sim \Gamma(m,1/m)$, where $m \ge 1$ is the {Nakagami-$m$ parameter}. The case $m=1$ corresponds to Rayleigh fading, while larger $m$ values represent less severe fading.

To capture the key aspects of wireless communication among UAVs, we consider the \textbf{fading-aware aggregate interference model}.  
At any receiver within the network, interference arises from concurrent transmissions by other UAVs that are active on the same channel.  
Assuming identical transmit power and independent fading across links, the aggregate interference power can be expressed as
\begin{equation}\label{eq:interferencemodel}
I \;=\; \sum_{i \in \Phi_{\mathrm{tx}}\setminus\{s\}} G_i\,U_i^{-\alpha},    
\end{equation}
where $\Phi_\mathrm{tx}$ denotes the set of concurrent transmitters, $U_i$ is the distance to the $i$-th transmitter, $G_i$ is its fading gain, and $s$ indexes the serving transmitter. 

This formulation is also referred to as a \emph{finite-user shot-noise process} widely used in stochastic geometry analyses of wireless networks.

\begin{figure}[t]
\centering
\begin{tikzpicture}[scale=0.5]
  
  \def\R{2.8}
  \def\ry{0.9}
  \def\H{5.0}

  \draw[thick] (-\R,0) arc[start angle=180,end angle=360,
                           x radius=\R,y radius=\ry];

  \draw[thick] (-\R,0)--(-\R,\H);
  \draw[thick] (\R,0)--(\R,\H);

  \draw[thick,fill=blue!6] (0,\H) ellipse[x radius=\R, y radius=\ry];

  \draw (0,\H) circle (0 pt);

  \coordinate (P1) at ( 0.55*\R, 0.78*\H);
  \coordinate (P2) at (-0.60*\R, 0.44*\H);
  \coordinate (P3) at ( 0.00*\R, 0.24*\H); 
  \fill (P1) circle (1.4pt);
  \fill (P2) circle (1.4pt);
  \fill (P3) circle (1.4pt);

  \draw[red,thick] (P3) circle (6pt);
  \node[red,font=\small,anchor=west] at ($(P3)+(0.9,0.3)$) {UAV};
  \draw[->,red,thick,>=Stealth] ($(P3)+(0.8,0.25)$) -- (P3);

  \draw[<->,>=Stealth] (0,\H) -- (\R,\H)
      node[midway,above=-1pt] {$R$};
  \draw[decorate,decoration={brace,raise=6pt}]
       (-\R-0.6,0) -- (-\R-0.6,\H)
       node[midway,left=6pt] {$H$};
  \node at (0,0.56*\H) {$\mathcal{C}$};
\end{tikzpicture}

\caption{A UAV network in a cylindrical volume $\mathcal{C}$ of height $H$ and radius $R$.}
\label{fig:cylinder_uav_final}
\end{figure}
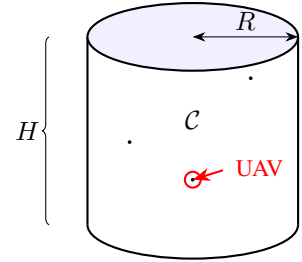

\subsection{Problem Formulation}
Based on the system model above, we evaluate the link reliability from the perspective of a typical receiver.
To formulate the analysis, we fix an arbitrary node as the typical receiver; because of the statistical symmetry, this node is representative of any receiver in the network.


For analytical tractability, we focus on the \emph{fully loaded} case in which all other UAVs transmit concurrently, i.e., $\Phi_{\mathrm{tx}}$ contains all nodes except the receiver, but the results can be easily extended to arbitrary $\Phi_{\mathrm{tx}} \subseteq [N]$.
This worst-case interference assumption yields a conservative (lower-bound) characterization of reliability.

Under identical transmit power and omnidirectional antennas, the receiver associates with the transmitter that maximizes the average received power, which is equivalent to the nearest node. 
Let $L_s$ be the resulting serving distance and $G_s$ the serving-link fading gain.

A key performance indicator in interference-limited networks is the \emph{coverage probability}, which quantifies the likelihood that the SIR (i.e., $\text{SIR} := \frac{G_s\,L_s^{-\alpha}}{I}$) at the typical receiver exceeds the decoding threshold $\beta$:
$$
P_c \;=\; \mathbb{P}\!\left(\frac{G_s\,L_s^{-\alpha}}{I} > \beta \right).
$$
Our objective is to derive a closed-form expression for $P_c$ as a function of the system parameters $(N, R, H, \alpha, m, \beta)$.

\section{Coverage Probability Analysis}
In this section, we derive an analytical expression for the coverage probability $P_c$.
Our derivation proceeds in four steps: we first characterize the distance distribution between UAVs, then obtain the distribution of the serving link, model the interference through its Laplace transform, and finally combine these results to evaluate the overall coverage probability.

\subsection{Distance Distribution between Any Two Nodes}
To capture the spatial relationship between UAVs uniformly deployed in the cylindrical region $\mathcal{C}$,
one of our key contributions is the \emph{probability density function (PDF)} of the Euclidean distance $L$ between any two randomly selected nodes, as shown in Theorem~\ref{thm:d_dist}.

\begin{theorem}\label{thm:d_dist}
Consider two nodes drawn independently and uniformly from a cylinder with base radius $R$ and height $H$. Let $L$ denote their Euclidean distance, and let $l$ be a realization of $L$.
Denote by $D_{xy}$ and $D_z$ their planar and vertical separations, respectively, and define the squared variables
$r = D_{xy}^2$ and $z = D_z^2$. 
Then, the PDF of $L$ is
\begin{equation*}
	f_L(l) = 2l \int_{0}^{l^2}  f_{XY}(r)\, f_Z(l^2 - r)\, dr,
\end{equation*}
where $f_{XY}(r)$ and $f_Z(z)$ represent the PDFs of the squared horizontal and vertical distance components, respectively, given by
\begin{equation*}
	f_{XY}(r) = \frac{f_{D_{xy}}(\sqrt{r})}{2\sqrt{r}}, 
	\qquad
	f_Z(z) = \frac{H - \sqrt{z}}{H^2 \sqrt{z}},
\end{equation*}
for $0 \le r \le 4R^2$ and $0 \le z \le H^2$.
\end{theorem}
	\begin{proof}[Proof Sketch]
		Due to the uniform distribution of points, $D_{xy}$ and $D_z$ are statistically independent random variables. In lemma~\ref{lm:disk} and \ref{lm:segment}, respectively, we know their respective PDFs, $f_{D_{xy}}(r)$ and $f_{D_z}(z)$.
		
		To find the PDF of $L$, we first consider the squared distances $D_{xy}^2$ and $D_z^2$. Since $D_{xy}$ and $D_z$ are independent, the PDF of their sum, $L^2 = D_{xy}^2+D_z^2$, can be obtained by the convolution of their individual PDFs. The final PDF of the distance $L$, $f_L(l)$, is then found through a variable transformation from the PDF of $L^2$. For detailed calculation information, please refer to the appendix.
	\end{proof}

Theorem~\ref{thm:d_dist} represents the overall distance distribution as a convolution of the independent planar and vertical components, forming the geometric foundation for the subsequent analysis.

The corresponding \emph{cumulative distribution function (CDF)} is defined by the integral: $F_L(l):= \mathbb{P}(L \leq l) = \int_0^l f_L(x) \, dx$.
Due to the complexity of $f_L(l)$, this integral generally does not have a simple closed-form solution and must be evaluated numerically.

\subsection{Distance Distribution from the Serving Node}
	We first characterize the distribution of the service distance $L_s$.
    Recall that $L_s$ is the minimum among the $N-1$ distances from the typical receiver to the transmitters, i.e., $ L_s = \min\{L_1, L_2, \dots, L_{N-1}\} $, where $L_i$ is the distance to the $i$-th transmitter. The set $\{L_i\}_{i=1}^{N-1}$ consists of i.i.d. random variables, each with PDF $f_L(l)$ and CDF $F_L(l)$. 
Using order statistics, we derive the PDF of $L_s$ as shown in Lemma~\ref{lm:ls}, denoted by $f_{L_s}$, based on the fading-aware aggregate interference model described in (\ref{eq:interferencemodel}).
	\begin{lemma}\label{lm:ls}
		The PDF $f_{L_s}(l)$ is given by 
        $$f_{L_s}(l) = (N-1) \left(1 - F_L(l)\right)^{N-2} f_L(l).$$
	\end{lemma} 
	\begin{proof}
		The event $\{L_s > l\}$ occurs if and only if all $N-1$ distances $L_i$ are greater than $l$.
		\begin{equation*}
        \begin{aligned}
			\mathbb{P}(L_s > l) &= \mathbb{P}(L_1 > l, \dots, L_{N-1} > l) \\
            & = \prod_{i=1}^{N-1} \mathbb{P}(L_i > l)  = \left(1 - F_L(l)\right)^{N-1},
        \end{aligned}
		\end{equation*}
		where the last two equalities are due to the distances $L_i$ being i.i.d. The lemma follows by taking derivatives of the above equation.
		

	\end{proof}

\subsection{Modeling the Aggregated Interference}
A key challenge in modeling the interference lies in accurately capturing the cumulative effect of multiple transmitters, 
each subject to independent fading and located at random distances in a three-dimensional space. 
Unlike simplified two-dimensional network models or those assuming infinite spatial domains, 
our finite cylindrical deployment introduces spatial boundary effects that complicate analytical tractability.  
Moreover, the interference terms are coupled through both distance-dependent path loss and random fading, 
making the exact distribution of the aggregate interference intractable.

To overcome these challenges, we employ the \emph{Laplace transform} of the aggregate interference power, 
which provides a tractable yet precise characterization of its distribution.  
This transform approach effectively captures the combined effects of the spatial geometry, path loss, and Nakagami-$m$ fading 
in a compact analytical expression—constituting one of the key steps in our framework.

Let the interference distance $U$ refer to the distance from the receiver to any one of the interfering (i.e., non-serving) transmitters, and denote by $f_{U|L_s}(u|l) $ the conditional PDF of $U$, given a service distance instantiation $L_s=l$. Then, given our proposed Lemma~\ref{lm:condistance} (see Appendix), the interference result is summarized in the following lemma.
	
	\begin{lemma}
		The Laplace Transform of the interference can be written as
		\begin{equation*}
			\mathcal{L}_{I}(t|l) = \left[ \int_{l}^{\infty} \left(1 + \frac{t u^{-\alpha}}{m}\right)^{-m} f_{U|L_s}(u|l) \, du \right]^{N-2},
		\end{equation*}
	\end{lemma}
	\begin{proof}
		The total interference is the sum of $N-2$ i.i.d. power components. Therefore, its Laplace transform is the product of the individual transforms, raised to the power of $N-2$.
		\begin{equation*}
			\mathcal{L}_{I}(t|l) = \left( \mathbb{E}_{U, G}\left[ e^{-t G U^{-\alpha}} | L_s=l \right] \right)^{N-2}.
		\end{equation*}
		First, we average over the interference channel gain $G$ (with parameters $m, 1/m$), which is equivalent to finding the \emph{moment generating function (MGF)} of a Gamma distribution:
		\begin{equation*}
			\mathbb{E}_{G}\left[ e^{-t G U^{-\alpha}} \right] = \left(1 + \frac{t U^{-\alpha}}{m}\right)^{-m}.
		\end{equation*}
		Next, we average over the random interference distance $U$ by integrating against its conditional PDF:
		\begin{equation*}
			\mathcal{L}_{I}(t|l) = \left[ \int_{l}^{\infty} \left(1 + \frac{t u^{-\alpha}}{m}\right)^{-m} f_{U|L_s}(u|l) \, du \right]^{N-2}.
		\end{equation*}
		\end{proof}

		\subsection{Deriving the Coverage Probability}
We now combine these results to derive the coverage probability of the typical receiver, as shown in Theorem~\ref{thm:main}.
		\begin{theorem}\label{thm:main}
			The coverage probability is
            \begin{equation*}
			\begin{aligned}
				P_c=&\int_0^{\sqrt{4R^2+H^2}}{\sum_{k=0}^{m-1}{\frac{(-t)^k}{k!}\left[ \frac{\partial ^k}{\partial t^k}\mathcal{L} _I(t|l) \right] _{t=m\beta l^{\alpha}}}} \nonumber
				\\
				&\times (N-1)\left( 1-F_L(l) \right) ^{N-2}f_L(l)\,dl.
			\end{aligned}                
            \end{equation*}
		\end{theorem}
		\begin{proof}
			The main idea is to first calculate the conditional coverage probability given a fixed service distance $L_s=l$, i.e., $\mathbb{P}(\text{SIR} > \beta | L_s=l)$. Then, according to the tower law and the fact that $\{L_s < 0\}$ has a zero measure, we have that
			\begin{equation*}
				P_c = \int_{0}^{\infty} \mathbb{P}(\text{SIR} > \beta | L_s=l) \cdot f_{L_s}(l) \, dl.
			\end{equation*} 
			By the definition of SIR, we have
			\begin{align*}
				\mathbb{P}(\text{SIR} > \beta | L_s=l) &= \mathbb{P}\left(\frac{G_s l^{-\alpha}}{I} > \beta | L_s=l \right)\\
				&= \mathbb{P}(G_s > \beta l^\alpha I | L_s=l) ,
			\end{align*}
			where $I = \sum_{i=1}^{N-2} G_i U_i^{-\alpha}$ is the total interference power.
			By the tower law over $I$, we have that
			\begin{equation*}
				\mathbb{P}(\text{SIR} > \beta | L_s=l) = \mathbb{E}_I\left[ \mathbb{P}(G_s > \beta l^\alpha I | L_s=l, I) \right] ,
			\end{equation*}
			where $\mathbb{P}(G_s > x)$ is the complementary CDF of a Gamma-distributed random variable. For integer values of $m$, it has the series expansion:
			\begin{equation*}
				\mathbb{P}(G_s > x) = e^{-m x} \sum_{k=0}^{m-1} \frac{(m x)^k}{k!}.
			\end{equation*} 
			
			By substituting $x = \beta l^\alpha I$ and manipulating the expression, the conditional coverage probability can be elegantly expressed in terms of the Laplace transform of the interference, $\mathcal{L}_I(t|l) = \mathbb{E}_I[e^{-tI}|L_s=l]$, and its derivatives. Using the property $\mathbb{E}[I^k e^{-tI} | l] = (-1)^k \frac{d^k}{dt^k} \mathcal{L}_I(t|l)$, we get:
			\begin{equation*}
				\mathbb{P}(\text{SIR} > \beta | L_s=l) = \sum_{k=0}^{m-1}\frac{(-t)^{k}}{k!}\left[\frac{\partial^{k}}{\partial t^{k}}\mathcal{L}_{I}(t|l)\right] _{t=m\beta l^{\alpha}}.
			\end{equation*}

			The theorem then follows by substituting all derived components back into the main integral.
		\end{proof}
		Theorem \ref{thm:main} provides an exact analytical expression for the coverage probability without relying on asymptotic or Poisson approximations. Unlike prior studies that offer approximate or simulation-based results, this formulation fully captures the joint effects of network geometry, finite node density, and Nakagami-$m$ fading. Hence, it represents the \textbf{first precise theoretical characterization} of coverage performance in the literature.
		
		\section{Simulation Results}
		In this section, we present numerical results to validate our theoretical framework and analyze the network performance. We first validate the derived PDF of the distance distribution, then analyze the coverage probability under various conditions, and finally demonstrate the superior accuracy of our model through a comparative analysis with the conventional \emph{Poisson point process~(PPP)} model.
		
		\subsection{Validation of the Distance Distribution}
One cornerstone of our analytical model is the accurate characterization of the Euclidean distance between two randomly located nodes within the cylinder. To validate the derived PDF, we perform Monte Carlo simulations by generating numerous random point pairs and computing the empirical distance distribution.

Figs.~\ref{fig:pdf_validation_flat} and~\ref{fig:pdf_validation_tall} compare the simulated histograms with the theoretical PDF for two representative geometries: a ``squat and short'' cylinder ($R=120, H=20$) and a ``slender and tall'' cylinder ($R=20, H=120$). In both cases, the analytical curves exhibit near-perfect agreement with the simulation results, confirming the correctness of the derivation and establishing a solid foundation for subsequent performance analysis.
		
\begin{figure}[htbp]
    \centering
    \begin{subfigure}[b]{0.8\linewidth}
        \centering
        \includegraphics[width=\linewidth]{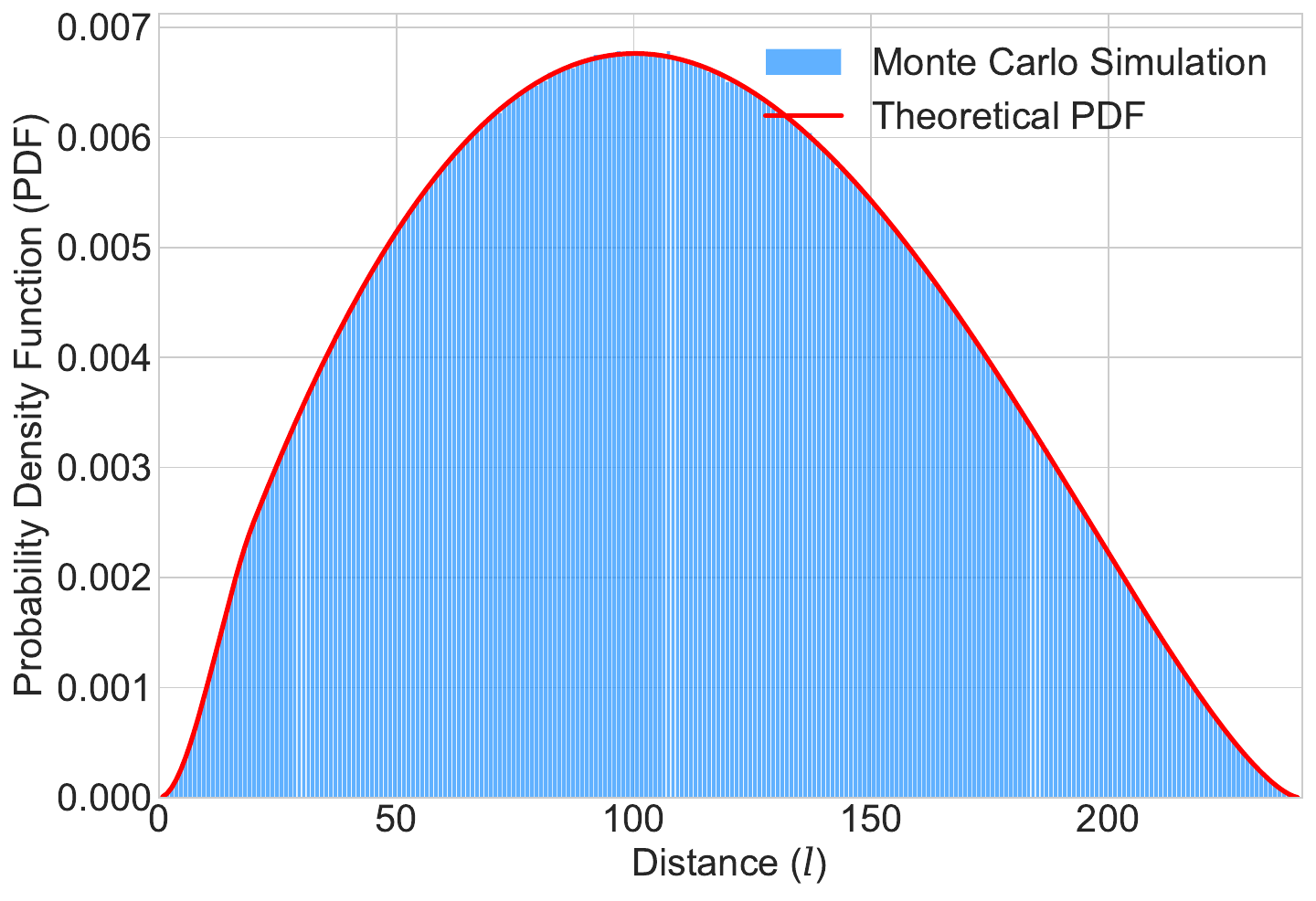}
        \caption{``Squat'' cylinder ($R=120$, $H=20$).}
        \label{fig:pdf_validation_flat}
    \end{subfigure}
    \hfill
    \begin{subfigure}[b]{0.8\linewidth}
        \centering
        \includegraphics[width=\linewidth]{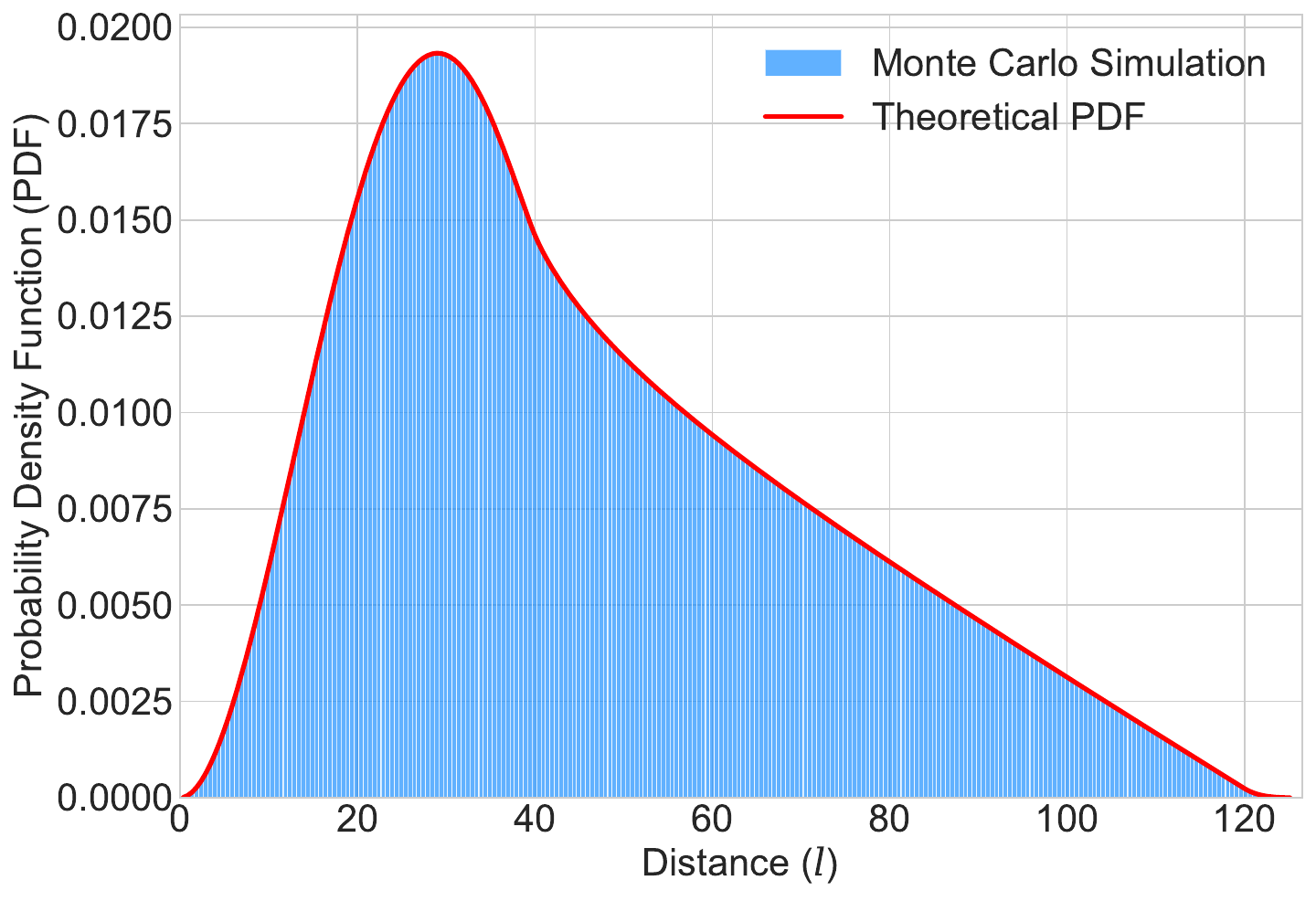}
        \caption{``Slender'' cylinder ($R=20$, $H=120$).}
        \label{fig:pdf_validation_tall}
    \end{subfigure}
    \caption{Validation of the theoretical PDF of the inter-node distance against Monte Carlo simulation results for two representative cylindrical geometries. The analytical curves show near-perfect agreement with simulation data.}
    \label{fig:pdf_validation_combined}
\end{figure}
		
		\subsection{Coverage Probability Analysis}
		First, we analyze the impact of the cylinder's aspect ratio on performance. Fig.~\ref{fig:coverage_vs_height} plots the coverage probability as a function of the number of UAVs for cylinders with a fixed radius $R$ but varying heights $H$. Each curve corresponds to a different cylinder height. The results clearly show that for a fixed number of nodes, the coverage probability degrades as the cylinder height increases. This is because a larger volume leads to greater average inter-node distances, which in turn weakens the received signal power and reduces the likelihood of successful communication links.
		
		\begin{figure}[htbp]
			\centering
			\includegraphics[width=0.8\linewidth]{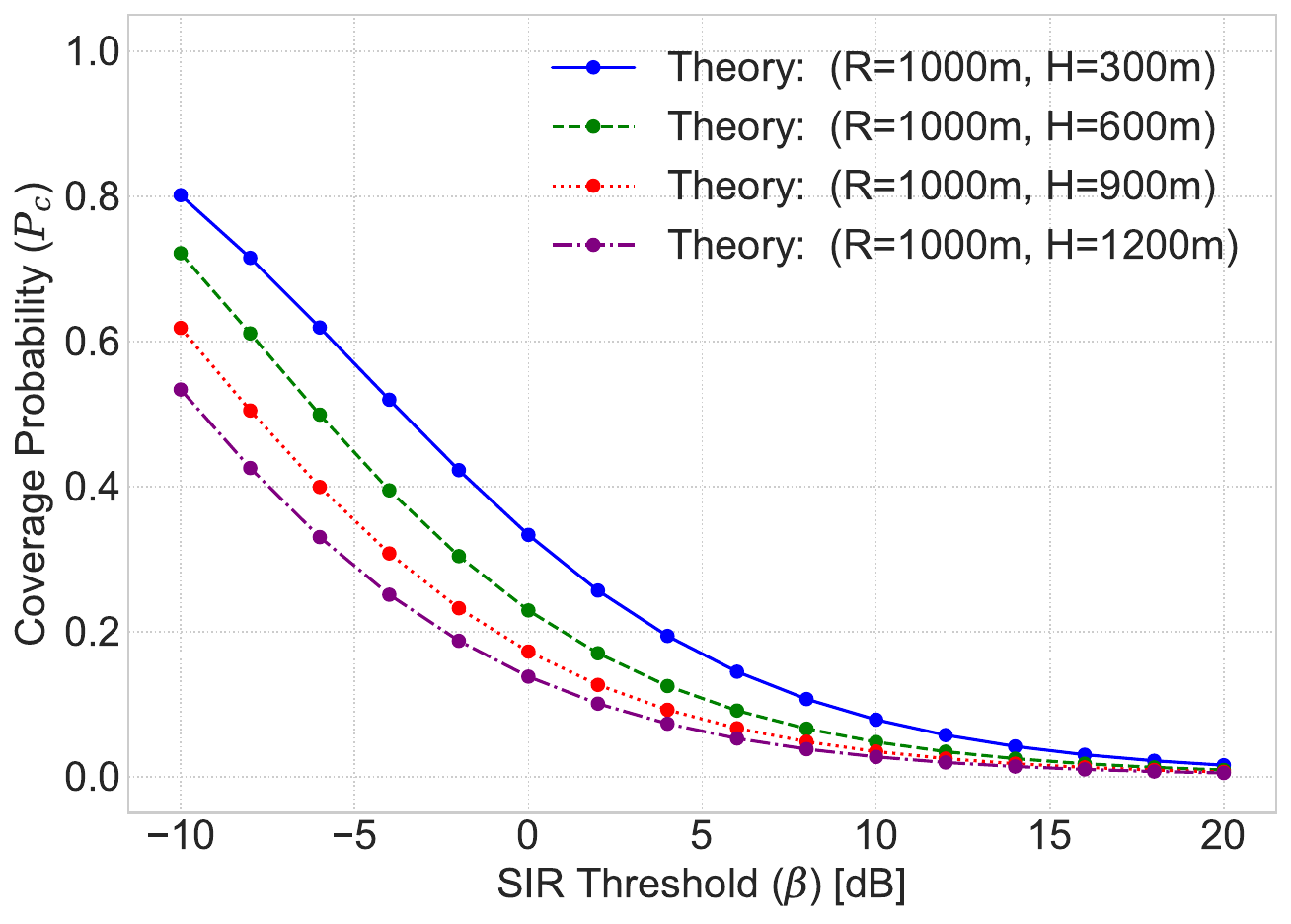}
			\caption{Coverage probability as a function of the number of UAVs for cylinders with varying heights.}
			\label{fig:coverage_vs_height}
		\end{figure}
		
		Next, we examine the influence of channel conditions. Fig.~\ref{fig:coverage_vs_fading} illustrates the coverage probability under different Nakagami-$m$ fading scenarios, where the parameter $m$ quantifies the fading severity. The case $m=1$ corresponds to Rayleigh fading (most severe), while larger values of $m$ represent less severe fading conditions. The plot indicates a significant improvement in coverage probability as $m$ increases from 1 to 3. This result aligns with theoretical expectations, confirming that more favorable channel conditions lead to a more robust and reliable network performance.
		
		\begin{figure}[htbp]
			\centering
			\includegraphics[width=0.8\linewidth]{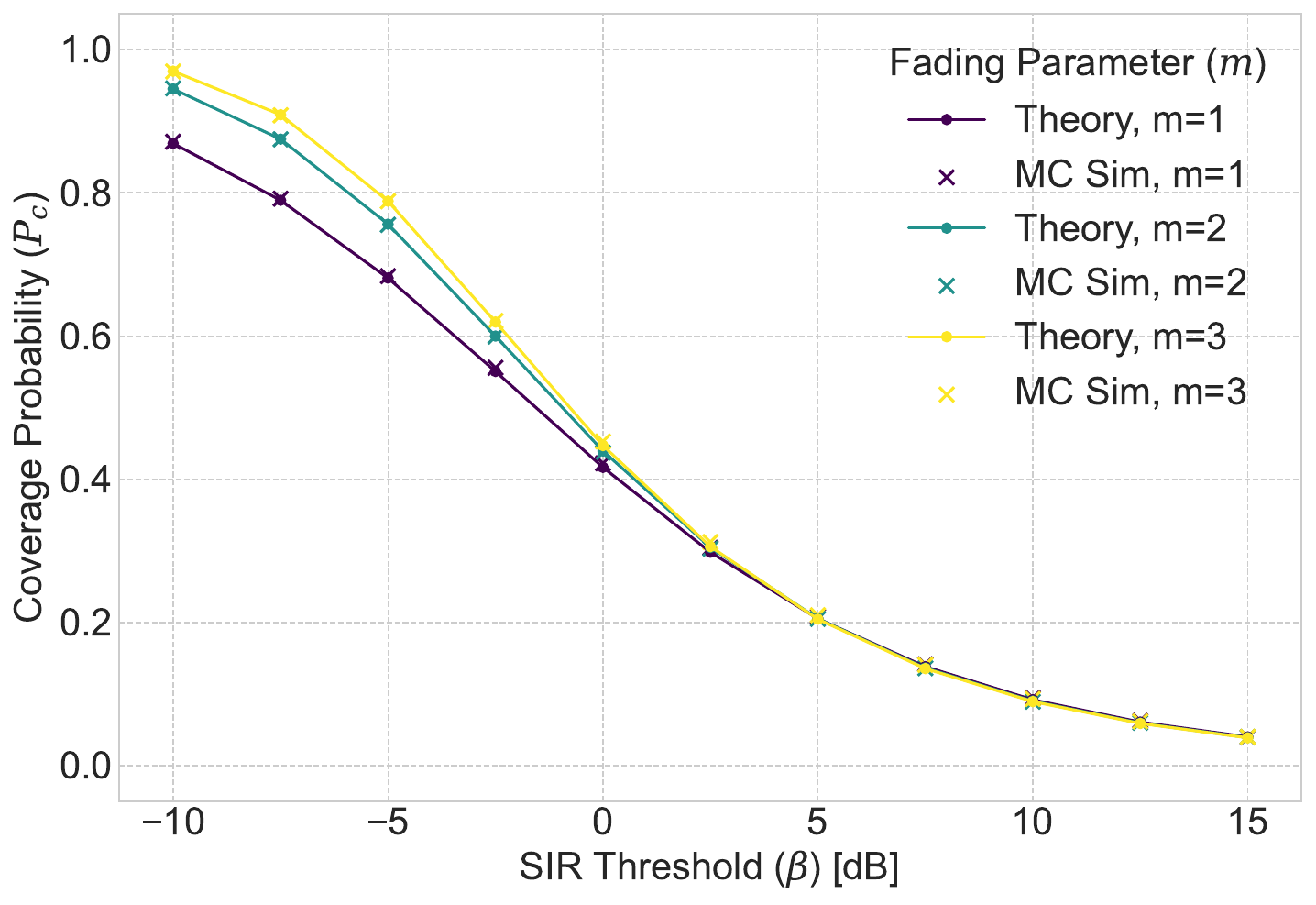}
			\caption{The impact of the Nakagami-$m$ fading parameter on network coverage probability.}
			\label{fig:coverage_vs_fading}
		\end{figure}

		\subsection{Comparative Analysis with the PPP Model}
		
		A primary contribution of this work is the development of a model that more accurately reflects a real-world deployment than the widely used PPP approximation. Our BPP-based model inherently accounts for two critical physical realities: a fixed number of nodes ($N$) and the presence of hard boundaries in a finite volume. In contrast, the PPP model assumes an infinite number of nodes distributed over an infinite space, thus neglecting boundary effects.
		
		To demonstrate the importance of these distinctions, Fig.~\ref{fig:comparison_ppp} presents a direct comparison between our proposed model, the PPP approximation, and the ground truth obtained from Monte Carlo simulation for the specified cylindrical network. The simulation result (blue curve) serves as the accurate benchmark. The analysis reveals two key observations:
		\begin{enumerate}
			\item Our proposed BPP-based theoretical model (red curve) perfectly aligns with the simulation results across the entire range of parameters, validating its high fidelity.
			\item The PPP-based approximation (green curve) exhibits a significant deviation from the ground truth. This discrepancy arises because the PPP model fails to capture the reduced interference experienced by nodes near the boundary (the boundary effect), leading to an inaccurate prediction of network performance.
		\end{enumerate}
		
		This comparison unequivocally demonstrates that for analyzing performance in finite, boundary-constrained scenarios, which are typical of UAV deployments, our BPP-based model offers substantially higher accuracy than the conventional PPP approximation.
		
		\begin{figure}[htbp]
			\centering
			\includegraphics[width=0.8\linewidth]{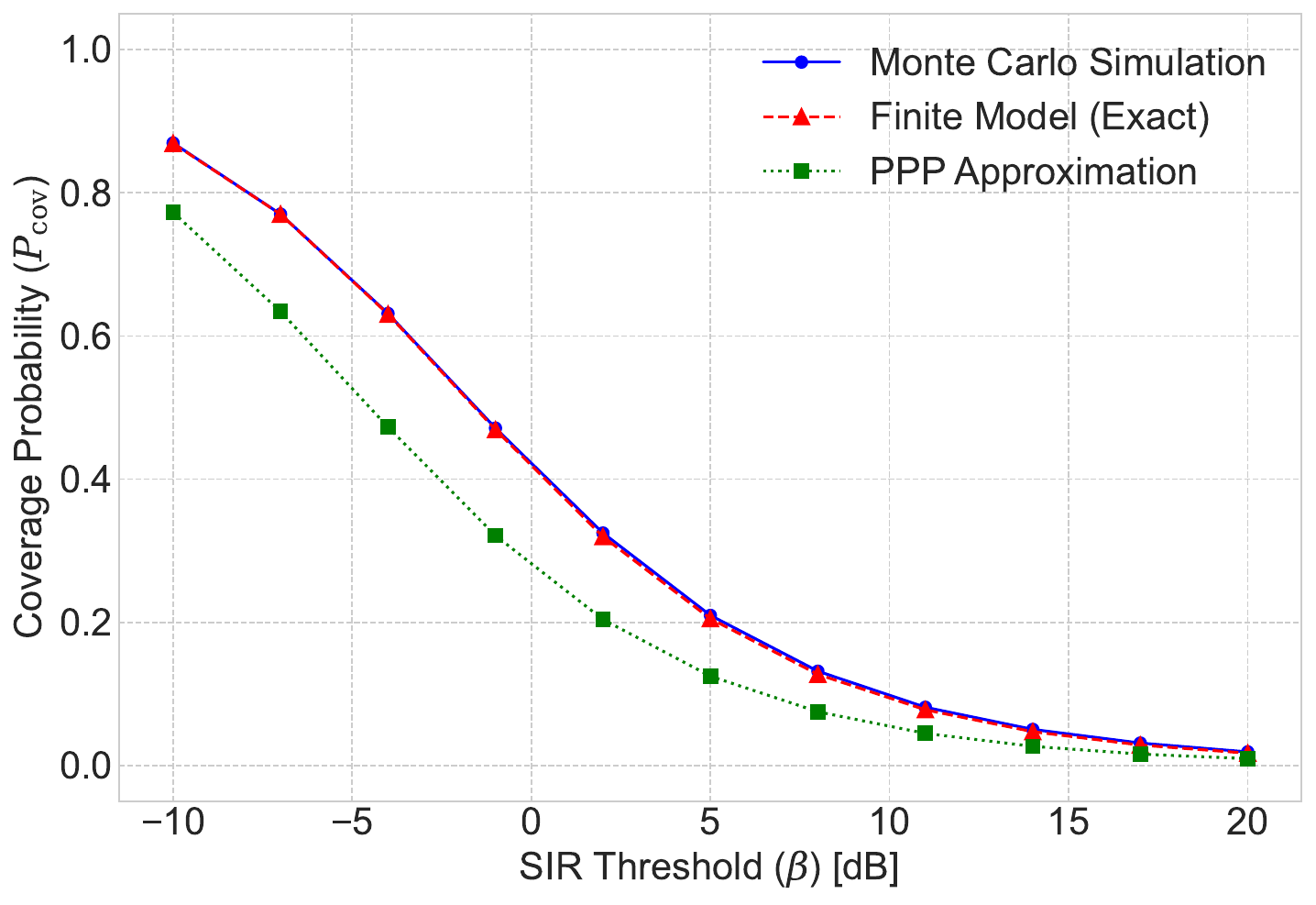}
			\caption{Comparison of coverage probability: Proposed BPP model vs. PPP approximation vs. Monte Carlo simulation (ground truth).}
			\label{fig:comparison_ppp}
		\end{figure}

		\section{CONCLUSION}
This paper resolves a long-standing open problem in stochastic geometry by deriving the exact probability density function of the Euclidean distance between two uniformly random points within a finite 3D cylinder. Building on this result, we established a unified analytical framework for evaluating the coverage probability of finite 3D wireless networks modeled by a BPP, incorporating path loss and Nakagami-$m$ fading. The framework overcomes the mathematical intractability that has long hindered finite-node analysis by leveraging independence structure, convolution geometry, and Laplace-transform derivatives, resulting in a formulation that is both mathematically exact and computationally efficient. Extensive Monte Carlo simulations confirm the precision of the theory and demonstrate substantial accuracy gains over traditional Poisson-based models.

Future work can build upon this foundation in several directions. Possible extensions include analyzing connectivity, capacity, and energy efficiency under the same geometric model, and exploring non-uniform node distributions, altitude-dependent channels, and multi-cylinder deployments. Applying this exact geometry to optimize UAV placement, trajectory design, or swarm coordination also represents a promising avenue for advancing real-world networked autonomous systems.


\onecolumn
\appendix 
The following two well-established geometric results from~\cite{8} are used as building blocks in our derivations. 
They provide the exact distance distributions between two uniformly random points within a finite disk 
and along a finite vertical line segment, respectively.
	\begin{lemma}[Theorem 2.3.13 in \cite{8}]\label{lm:disk}
		The distance distribution PDF of two uniformly random points in a finite disk of radius R is given as 
		\begin{equation}
			f_{D_{xy}}(v) = \frac{4v}{\pi R^2} \left( \arccos\left(\frac{v}{2R}\right) - \frac{v}{2R} \sqrt{1 - \left(\frac{v}{2R}\right)^2} \right),
		\end{equation}
		for $0 \le v \le 2R $.
	\end{lemma}
	
	\begin{lemma}[Theorem 2.2.2 in \cite{8}]\label{lm:segment}
		The probability density function of the distance between two uniformly random points within a vertical line segment of length H is given as 
		\begin{equation}
			f_{D_z}(z) = \frac{2(H-z)}{H^2},
		\end{equation}
		for $ 0 \le z \le H$.
	\end{lemma}

	\begin{lemma}\label{lm:condistance}
		Given an instantiation of $L_s=l$, the conditional PDF of $U$ is:
		\begin{equation}
			f_{U|L_s}(u|l) = \frac{f_L(u)}{1 - F_L(l)}, \quad u \ge l,
		\end{equation} 
	\end{lemma}
	
    \begin{proof}
The conditional PDF of an interferer's distance $U$ given $L_s=l$ is derived from the fact that all interferer distances must be greater than $l$. By definition of conditional probability, for any $u \ge l$, the conditional CDF is:
\begin{align*}
    \mathbb{P}(U \le u | L_s = l) &= \mathbb{P}(L_i \le u | L_i > l) \\
    &= \frac{\mathbb{P}(l < L_i \le u)}{\mathbb{P}(L_i > l)} \\
    &= \frac{F_L(u) - F_L(l)}{1 - F_L(l)}.
\end{align*}
Taking the derivative with respect to $u$ yields the conditional PDF:
$$
f_{U|L_s}(u|l) = \frac{d}{du}\left( \frac{F_L(u) - F_L(l)}{1 - F_L(l)} \right) = \frac{f_L(u)}{1 - F_L(l)},
$$
for $u \ge l$, which completes the proof.
\end{proof}

	 \begin{corollary}
     By Theorem~\ref{thm:d_dist}, we can obtain the closed-form expression of the distance distribution PDF.
     For the domain of integration depends on the relationship between $h$, $R$ and $l$, the PDF is divided into 4 part.
	
 	When $0\le l\le 2R,l\le h$, the final PDF expression as follows 
		\begin{equation}
	 		\begin{aligned}
	 			f_L\left( l \right) =&\frac{2l^2(2h-l)}{R^2h^2}+\frac{l^2(l^2+2R^2)\sqrt{4R^2-l^2}}{2\pi R^4h^2}\\
	 			&+\frac{4l(l^2-R^2)}{\pi R^2h^2}\mathrm{arc}\sin \left( \frac{l}{2R} \right)\\
	 			&+\frac{32l}{3\pi Rh}\left( 1-\frac{l^2}{4R^2} \right) K\left( \frac{l}{2R} \right) \\
 			&-\left( 1+\frac{l^2}{4R^2} \right) E\left( \frac{l}{2R} \right)
	 		\end{aligned}
	 	\end{equation}

	 	When $0\ \le 2R < l ,l \leq h$, the final PDF expression as follows
	 	\begin{equation}
	 		\begin{aligned}
	 			f_L(l)
	 			& = \frac{4l^2}{R^2 h} - \frac{2l}{h^2} + \frac{4l^2(l^2-4R^2)}{3\pi R^4 h} K\left(\frac{2R}{l}\right) \\
	 			&- \frac{4l^2(l^2+4R^2)}{3\pi R^4 h} E\left(\frac{2R}{l}\right)\\
	 		\end{aligned}
            \end{equation}

	 	when $0 \leq l \leq 2R,h < l$, the final PDF expression as follows 
	 	\begin{figure*}[t]
	 		
	 		\begin{equation}		
	 			\begin{aligned}
	 				f_L(l) &=\left( \frac{8l}{\pi R^2h}\left\{ h\cdot \mathrm{arc}\cos \left( \frac{\sqrt{l^2-h^2}}{2R} \right) -\left( \frac{l^2}{2R}-2R \right) A_3(l,R,h)-2RB_3(l,R,h) \right\} \right)\\
	 				&+\left( -\frac{4l}{\pi R^2h^2}\left\{ \begin{array}{c}
	 					\left( l^2-2R^2 \right) \mathrm{arc}\cos \left( \frac{l}{2R} \right) -\left( l^2-h^2-2R^2 \right) \mathrm{arc}\cos \left( \frac{\sqrt{l^2-h^2}}{2R} \right)\\
	 					-\frac{l\sqrt{4R^2-l^2}}{2}+\frac{\sqrt{(l^2-h^2)(4R^2-l^2+h^2)}}{2}\\
	 				\end{array} \right\} \right)\\
	 				&+\left( -\frac{16l}{3\pi Rh}\left\{ \left( \frac{l^2}{2R^2}-1 \right) \left[ C_3(l,R,h) \right] +\left( 1-\frac{l^2}{4R^2} \right) \left[ D_3(l,R,h) \right] +\frac{h\sqrt{(l^2-h^2)(4R^2-l^2+h^2)}}{8R^3} \right\} \right)\\
	 				&+\left( \frac{2l}{\pi h^2}\left[ \begin{array}{c}
	 					\mathrm{arc}\sin \left( \frac{l^2}{2R^2}-1 \right)\\
	 					-\mathrm{arc}\sin \left( \frac{l^2-h^2-2R^2}{2R^2} \right)\\
	 				\end{array} \right] +\frac{l}{\pi R^4h^2}\left[ \begin{array}{c}
	 					\frac{l(l^2-2R^2)}{2}\sqrt{4R^2-l^2}\\
	 					-\frac{l^2-h^2-2R^2}{2}\sqrt{(l^2-h^2)(4R^2-l^2+h^2)}\\
	 				\end{array} \right] \right)\\
	 			\end{aligned}
	 		\end{equation}
	 	\end{figure*}
	 	where
	 	$$A_3(l,R,h)= K\left( \frac{l}{2R} \right) -F\left( \mathrm{arc}\cos \left( \frac{h}{l} \right) ,\frac{l}{2R} \right)$$
	 	$$B_3(l,R,h)= E\left( \frac{l}{2R} \right) -E\left( \mathrm{arc}\cos \left( \frac{h}{l} \right) ,\frac{l}{2R} \right)$$
	 	$$C_3(l,R,h)=E\left( \frac{l}{2R} \right) -E\left( \mathrm{arc}\sin \left( \frac{\sqrt{l^2-h^2}}{l} \right) ,\frac{l}{2R} \right)$$
	 	$$D_3(l,R,h)=K\left( \frac{l}{2R} \right) -F\left( \mathrm{arc}\sin \left( \frac{\sqrt{l^2-h^2}}{l} \right) ,\frac{l}{2R} \right)$$

	 	When $0\ \le 2R < l ,h < l$, the final PDF expression as follows
	 	\begin{equation}
	 		\begin{aligned}
	 			f_L\left( l \right) = & - \frac{l(3l^2+6R^2+h^2)}{6\pi R^4 h^2}\sqrt{(l^2-h^2)(4R^2-l^2+h^2)} \\
	 			& + \frac{4l(l^2+h^2)}{\pi R^2 h^2}  \arccos\left(\frac{\sqrt{l^2-h^2}}{2R}\right)- \frac{2l}{h^2} \\
	 			& + \frac{4l}{\pi h^2} \arcsin\left(\frac{\sqrt{l^2-h^2}}{2R}\right) \\
	 			& - \frac{4l^2(4R^2-l^2)}{3\pi h R^4}A_4(l,R,h)  \\
	 			& - \frac{4l^2(l^2+4R^2)}{3\pi h R^4}B_4(l,R,h)  \\
	 		\end{aligned}
	 	\end{equation}

	 	where
	 	$$A_4(l,R,h) = \left[ K\left(\frac{2R}{l}\right) - F\left(\arcsin\left(\frac{\sqrt{l^2-h^2}}{2R}\right), \frac{2R}{l}\right) \right]$$
		
	 	$$B_4(l,R,h) = \left[ E\left(\frac{2R}{l}\right) - E\left(\arcsin\left(\frac{\sqrt{l^2-h^2}}{2R}\right), \frac{2R}{l}\right) \right]$$
	 \end{corollary}
\bibliographystyle{ieeetr}
\bibliography{references}

\begin{thebibliography}{1}

\bibitem{3}
P.~K. Sharma and D.~I. Kim, ``{Random 3D mobile UAV networks: Mobility modeling and coverage probability},'' {\em IEEE Transactions on Wireless Communications}, vol.~18, no.~5, pp.~2527--2538, 2019.

\bibitem{4}
C.~K. Armeniakos, P.~S. Bithas, and A.~G. Kanatas, ``{Finite point processes in a truncated octahedron-based 3D UAV network},'' {\em IEEE Transactions on Vehicular Technology}, vol.~71, no.~7, pp.~7230--7243, 2022.

\bibitem{5}
M.~Mozaffari, W.~Saad, M.~Bennis, and M.~Debbah, ``{Unmanned aerial vehicle with underlaid device-to-device communications: Performance and tradeoffs},'' {\em IEEE Transactions on Wireless Communications}, vol.~15, no.~6, pp.~3949--3963, 2016.

\bibitem{6}
V.~V. Chetlur and H.~S. Dhillon, ``{Downlink coverage analysis for a finite 3-D wireless network of unmanned aerial vehicles},'' {\em IEEE Transactions on Communications}, vol.~65, no.~10, pp.~4543--4558, 2017.

\bibitem{7}
X.~Shi and N.~Deng, ``{Modeling and analysis of mmWave UAV swarm networks: A stochastic geometry approach},'' {\em IEEE Transactions on Wireless Communications}, vol.~21, no.~11, pp.~9447--9459, 2022.

\bibitem{9}
F.~Tong and J.~Pan, ``Random-to-random nodal distance distributions in finite wireless networks,'' {\em IEEE Transactions on Vehicular Technology}, vol.~66, no.~11, pp.~10070--10083, 2017.

\bibitem{10}
R.~Pure, S.~Durrani, F.~Tong, and J.~Pan, ``Distance distribution between two random points in arbitrary polygons,'' {\em Mathematical Methods in the Applied Sciences}, vol.~45, no.~5, pp.~2760--2775, 2022.

\bibitem{8}
A.~M. Mathai, {\em {An introduction to geometrical probability: Distributional aspects with applications}}, vol.~1.
\newblock CRC Press, 1999.

\end{thebibliography}
\end{document}